\documentclass[11pt]{article}
\usepackage{fullpage}

\usepackage[dvipsnames]{xcolor}
\usepackage[pdftex,bookmarks=true,pdfstartview=FitH,colorlinks,linkcolor=Cerulean,filecolor=Cerulean,citecolor=Cerulean,urlcolor=Cerulean,backref=page]{hyperref}

\usepackage{amsmath}
\usepackage{amsfonts}

\usepackage{tikz}

\usepackage{amssymb}
\usepackage{physics}
\usepackage{mdframed}
\usepackage{xspace}
\usepackage{caption}
\usepackage{bbm}
\usepackage[
	lambda,
	operators,
	advantage,
	sets,
	adversary,
	landau,
	probability,
	notions,	
	logic,
	ff,
	mm,
	primitives,
	events,
	complexity,
	asymptotics,
	keys]{cryptocode}

\usepackage{amsthm}
\usepackage{cleveref}

\crefname{corollary}{corollary}{corollaries}
\Crefname{Corollary}{Corollary}{Corollaries}

\newtheorem{definition}{Definition}
\newtheorem{theorem}{Theorem}
\newtheorem{informaltheorem}{Informal Theorem}
\newtheorem{corollary}[theorem]{Corollary}
\newtheorem{informalcorollary}[informaltheorem]{Informal Corollary}
\newtheorem{lemma}{Lemma}

\newtheorem*{remark*}{Remark}

\newtheorem{fact}{Fact}

\newcommand{\R}{\mathbb{R}}
\newcommand{\N}{\mathbb{N}}

\newcommand{\calN}{\mathcal{N}}
\newcommand{\calY}{\mathcal{Y}}
\newcommand{\Ber}{\mathsf{Ber}}

\newcommand{\from}{\leftarrow}

\newcommand{\KeyGen}{\mathsf{KeyGen}}

\newcommand{\Enc}{\mathsf{Enc}}
\newcommand{\Dec}{\mathsf{Dec}}

\DeclareMathOperator*{\E}{\mathbb{E}}

\DeclarePairedDelimiterX{\inner}[2]{\langle}{\rangle}{#1, #2}

\let\leq\leqslant
\let\geq\geqslant
\let\le\leq
\let\ge\geq
\let\epsilon\varepsilon
\let\eps\epsilon
\let\draw\from

\renewcommand{\secpar}{\lambda}
\newcommand{\zo}{\{0,1\}}
\newcommand{\pred}[1]{\ensuremath{\mathsf{#1}}\xspace}
\newcommand{\SD}[1]{\ensuremath{\pred{SD}\left({#1}\right)}\xspace}

\newcommand{\cA}{\mathcal{A}}
\newcommand{\cB}{\mathcal{B}}

\newcommand{\cD}{\mathcal{D}}
\newcommand{\cE}{\mathcal{E}}
\newcommand{\cH}{\mathcal{H}}

\newcommand{\cN}{\mathcal{N}}
\newcommand{\cO}{\mathcal{O}}
\newcommand{\cP}{\mathcal{P}}

\newcommand{\cR}{\mathcal{R}}
\newcommand{\cS}{\mathcal{S}}
\newcommand{\cU}{\mathcal{U}}

\newcommand{\cX}{\mathcal{X}}
\newcommand{\cY}{\mathcal{Y}}

\newcommand{\KL}[2]{\mathrm{D}_{\mathrm{KL}}{\left({#1}\Vert{#2}\right)}}

\newcommand{\Bad}{\mathsf{Bad}}

\renewcommand{\entropy}[1]{H\left({#1}\right)}

\renewcommand{\paragraph}[1]{\medskip\noindent\textbf{#1}~}

\title{Black-Box Crypto is Useless for Pseudorandom Codes}
\author{
    Sanjam Garg\thanks{UC Berkeley and Exponential Science Foundation. Email: \texttt{sanjamg@berkeley.edu}. Author is supported in part by the AFOSR Award FA9550-24-1-0156 and research grants from the Bakar Fund, Supra Inc., and the Stellar Development Foundation.}
    \and Sam Gunn\thanks{UC Berkeley. Email: \texttt{gunn@berkeley.edu}.}
    \and Mingyuan Wang\thanks{NYU, Shanghai. Email: \texttt{mingyuan.wang@nyu.edu}.}
}
\date{June 2025}

\begin{document}

\maketitle

\begin{abstract}
    A \emph{pseudorandom code} is a keyed error-correction scheme with the property that any polynomial number of encodings appear random to any computationally bounded adversary.
    We show that the pseudorandomness of any code tolerating a constant rate of random errors cannot be based on black-box reductions to almost any generic cryptographic primitive: for instance, anything that can be built from random oracles, generic multilinear groups, and virtual black-box obfuscation.
    Our result is \emph{optimal}, as Ghentiyala and Guruswami (2024) observed that pseudorandom codes tolerating any sub-constant rate of random errors exist using a black-box reduction from one-way functions.
    
    The key technical ingredient in our proof is the hypercontractivity theorem for Boolean functions, which we use to prove our impossibility in the random oracle model.
    It turns out that this easily extends to an impossibility in the presence of ``crypto oracles,'' a notion recently introduced---and shown to be capable of implementing all the primitives mentioned above---by Lin, Mook, and Wichs (EUROCRYPT 2025).
\end{abstract}

\thispagestyle{empty}
\newpage
{\small\tableofcontents}
\thispagestyle{empty}
\newpage

\section{Introduction}
\label{sec:intro}

%%% Everything is in notes.tex. I don't think I deleted anything.
%% The old intro seemed tailored to the old result. Like, the question was ``do PRCs belong in minicrypt?''
%% Accordingly I have changed quite a bit.

A \emph{pseudorandom code} (PRC) is a keyed error-correction scheme whose codewords appear pseudorandom to any computationally bounded adversary who doesn't know the key.
In other words, it is a secret-key pseudorandom encryption scheme with the additional requirement that the decoding algorithm functions even if the transmission is corrupted.
Unless otherwise specified, a PRC should be robust to a \emph{constant rate} of random errors.

PRCs were defined by \cite{C:ChrGun24}, where they were shown to be equivalent to robust and undetectable watermarking for large language models.
% In watermarking, the goal is to plant signals in an algorithm's output which can only be detected by parties who hold the watermarking key.
% At a high level, a PRC is useful for watermarking because it makes it possible to replace the randomness in an AI pipeline with \emph{robust} pseudorandomness.
% Robustness of the PRC translates to robustness of the watermark, and pseudorandomness of the PRC translates to undetectability of the watermark (i.e., the quality of the generated content is preserved).
%%%% I tried to incorporate the watermarking discussion Mingyuan had put, but I really think it just distracts from the flow. ---Sam
These ideas have since been used for practical watermarking of AI-generated images and videos \cite{GZS25,YZC+25,HLLL25}.

Unfortunately, despite several constructions having been proposed \cite{C:ChrGun24,EPRINT:GheGur24,GM24}, every known PRC either suffers from quasipolynomial-time distinguishing attacks or relies on ad-hoc hardness assumptions.
It is therefore natural to ask,
\begin{center}
    \em
    % is there a fundamental barrier making it so difficult to construct a (constant rate) pseudorandom code from generic cryptographic primitives (such as one-way functions, public-key encryption, etc.)?
        %%%% Issue: ``constant rate'' would typically mean ``constant information rate'' in coding theory. Particularly considering that we already said we were considering a constant error rate in our first paragraph, I think we can't use ``constant rate'' here. Instead we could do:
    % is there a fundamental barrier to constructing a pseudorandom code, tolerating a constant rate of random errors, from generic cryptographic primitives (such as one-way functions, public-key encryption, etc.)?
        %%%% This is fine, but I think it's unnecessary to specify the error rate; that would make for specifying it 5 times on this first page alone. We say in our opening paragraph that we assume a constant rate of random errors, then we say again in the first sentence of 1.1, again in Informal Theorem 1, and once more after Informal Theorem 1.
        % And it makes this question sound much less natural to include it. I want the reader to be thinking ``PRC means it handles a constant rate of errors. They show a very-complete black-box separation for PRCs'' and not ``PRC means it handles some error channel. In particular a pseudorandom encryption scheme is a PRC. They show that in some regime of the error rate, there is a strong black-box separation.''
    is there a fundamental barrier to constructing a pseudorandom code from generic cryptographic primitives (such as one-way functions, public-key encryption, etc.)?
        %%%% I think this is the way to put it. It's not even ``incorrect'' in any way, as it's a question, and the question is sufficiently vague to include multiple types of PRCs; also because we specified only a few sentences earlier that ``PRC'' means constant rate of errors.
\end{center}

\subsection{Our contribution}
\label{sec:contrib}

We answer this question in the affirmative by showing that the black-box use of virtually all cryptographic primitives cannot yield a PRC resilient to a constant fraction of errors.

\begin{informaltheorem}
    Relative to a random oracle, there do not exist statistically-secure pseudorandom codes tolerating any constant rate of random bit-flip errors.\footnote{
    Random errors are the weakest error model considered in the literature. Since we are proving a lower bound, this only makes our result stronger.}
\end{informaltheorem}
There is a simple construction of pseudorandom codes tolerating any sub-constant error rate in the information-theoretic random oracle model \cite{EPRINT:GheGur24}.
Therefore, our result is essentially optimal in terms of the error rate.

In fact, this theorem actually implies much stronger separation results.
In the secret-key setting, the encoder and decoder effectively share a \emph{secret random oracle} which is not accessible to the pseudorandomness adversary.\footnote{
    Note that the encoder and decoder can prepend all of their queries to the random oracle $\cR$ with the shared secret key $\sk$. Then $\cR(\sk\Vert\cdot)$ is the secret random oracle.}
The ideas of \cite{LMW25} enable us to use this fact to prove our stronger separation.

That work introduced \emph{crypto oracles} to prove a strong black-box impossibility result for doubly efficient private information retrieval.
A crypto oracle (refer to \Cref{subsec:techo-crypto-oracle} or \Cref{sec:cryptooracle}) is a stateless algorithm $\cB^\cR$ with access to a secret random oracle $\cR$.
The work of \cite{LMW25} showed that crypto oracles can implement effectively all primitives in cryptography, including idealized primitives such as virtual black-box obfuscation and generic multilinear groups.

\begin{informalcorollary}
    Relative to any crypto oracle, there do not exist statistically secure pseudorandom codes tolerating any constant rate of random bit-flip errors.
\end{informalcorollary}

It is therefore not possible to build a PRC using black-box reductions to almost any generic cryptographic primitive.

By standard techniques in black-box separation literature~\cite{STOC:ImpRud89,TCC:ReiTreVad04}, our result immediately yields a black-box separation between secret-key PRCs tolerating a constant rate of errors and any primitives implied by crypto oracles.\footnote{
    Technically, our theorem shows that pseudorandom codes resilient to a constant rate of random errors do not exist relative to any crypto oracle $\cB^\cR$ and a $\mathsf{PSPACE}$ oracle. Similar to prior results, this corollary is established by observing that, relative to $\cB^\cR$ and $\mathsf{PSPACE}$, for any primitive $\cP$ implied by $\cB^\cR$, $\cP$ exists but PRCs do not exist.}

\begin{informalcorollary}
    Pseudorandom codes resilient to a constant rate of random errors are black-box separated from all primitives implied by crypto oracles, including random oracles, virtual black-box obfuscation, and generic multilinear groups.
\end{informalcorollary}

\subsection{Concurrent work}
An independent and concurrent work also demonstrates a black-box impossibility result for pseudorandom codes \cite{DMR25}.
Our works use similar arguments involving hypercontractivity to prove the impossibility in the random oracle model.
But whereas we go on to prove our impossibility relative to any crypto oracle, they introduce an interesting new class of oracles they call \emph{local oracles} and prove their impossibility relative to this class.
A local oracle is, very roughly, one for which most outputs are unrelated to each other.
As far as we know, these two oracle classes are incomparable to each other.

\section{Technical overview}
\label{sec:overview}
If the reader is already familiar with the prior work on pseudorandom codes, it should be possible to start at \Cref{subsec:techo-strategy}.

We begin this overview by recalling the definition of a pseudorandom code (PRC) and introducing useful notation in \Cref{subsec:techo-definitions}.
For a more complete discussion of PRC definitions, see, e.g. \cite{C:ChrGun24,GM24,EPRINT:GheGur24,AACDG25}.
We then reproduce an argument of Ghentiyala and Guruswami \cite{EPRINT:GheGur24} showing that there exist pseudorandom codes for any $\smallO{1}$ rate of random errors, assuming just the existence of one-way functions.

Our impossibility proof is outlined in \Cref{subsec:techo-strategy,subsec:techo-remove-one-query,subsec:techo-crypto-oracle}.
In \Cref{subsec:techo-strategy} we explain our proof strategy, which is to start by ruling out constructions in the random oracle model.
We show how to compile out the random oracle from any PRC, yielding a statistically-secure PRC in the standard model---something that is easily seen to be impossible.
We outline our analysis of this compiler in \Cref{subsec:techo-remove-one-query}, which is the key technical component of our proof.
Finally, we explain how this extends to a black-box separation from most of cryptography in \Cref{subsec:techo-crypto-oracle}, using an idea of Lin, Mook, and Wichs \cite{LMW25}.

\subsection{Definitions}
\label{subsec:techo-definitions}
In this work, we are interested in pseudorandom codes with robustness to the binary symmetric channel, which introduces random bit-flips.
Following the convention in Boolean analysis, we denote this channel by $\cN_{\rho}$.
For $\rho \in [0,1)$ and $x \in \{0,1\}^n$, $\cN_{\rho}(x)$ replaces each bit of $x$ with a random bit with probability $1-\rho$.
Note that $\cN_{\rho}$ is the binary symmetric channel $\mathsf{BSC}_{(1-\rho)/2}$.
This is the weakest model of noise considered in the literature for pseudorandom codes.

For the purposes of our overview, a \emph{pseudorandom code} (or PRC) for $\cN_{\rho}$ will be a pair of polynomial-time randomized algorithms $\Enc, \Dec$ that satisfy the following properties:
\begin{itemize}
    \item \textbf{Completeness / robustness}: If $\sk \from \{0,1\}^\secpar$ is random, $c \from \Enc(\sk)$, and $\tilde{c} \from \cN_{\rho}(c)$, then $\Dec(\sk, c) = 1$ with probability $1-\negl$.
    \item \textbf{Soundness}: For any fixed $x \in \{0,1\}^n$, we have $\Pr_{\sk \from \{0,1\}^\secpar}[\Dec(\sk, x) = \bot] \ge 1-\negl$.
    \item \textbf{Pseudorandomness}: For any polynomial-time adversary $\cA$ and any $m = \poly$,
    \[
        \abs{\Pr_{\substack{c_1, \dots, c_m \from \Enc(\sk) \\ \sk \from \{0,1\}^\secpar}}[\cA(c_1, \dots, c_m) = 1] - \Pr_{x_1, \dots, x_m \from \{0,1\}^n}[\cA(x_1, \dots, x_m) = 1]} \le \negl.
    \]
\end{itemize}
This definition is only the \emph{secret-key, zero-bit} case of the more general definition of \cite{C:ChrGun24}.
Since we are proving an impossibility, it only strengthens our result to consider this special case.

\subsection{Low-noise pseudorandom codes from pseudorandom functions}
\label{subsec:techo-low-noise-construction}
Before we get to our impossibility, let us recall how pseudorandom codes for $\cN_{\rho}$ can be built from pseudorandom functions (PRFs) if $\rho = 1 - \smallO{1}$.
This argument is from \cite[Section~3]{EPRINT:GheGur24}.

Let $\ell = (\log \secpar) / (1-\rho)$ and $\prf_\sk:\zo^\ell \to \zo^\ell$ be a PRF.
Our encoder $\Enc(\sk)$ first samples random strings $r_1, \dots, r_{\secpar^2} \from \{0,1\}^\ell$, then outputs
\[
    c = \left(r_1 \Vert \prf_\sk(r_1) \Vert \cdots \Vert r_{\secpar^2} \Vert \prf_\sk(r_{\secpar^2})\right).
\]
The decoder, upon receiving $x = \left(\widetilde{r}_1\Vert\widetilde{y}_1\Vert\cdots\Vert \widetilde{r}_{\secpar^2}\Vert\widetilde{y}_{\secpar^2}\right)$, simply checks if $\widetilde{y}_i = \prf_\sk(\widetilde{r}_i)$ for any $i$.

Pseudorandomness follows from the security of the PRF and the fact that $\ell = \omega(\log \secpar)$ if $\rho = 1-\smallO{1}$.
For soundness, suppose that $x$ is independent of $\sk$:
Then for any $i$, the probability that $\widetilde{y}_i = \prf_\sk(\widetilde{r}_i)$ is at most $2^{-\ell} + \negl$ by security of the PRF.
If $\rho = 1-o(1)$ then this is $\negl$, and by a union bound so is the probability that the decoder outputs $1$.

For robustness to $\cN_{\rho}$, suppose that $c \from \Enc(\sk)$ and $\tilde{c} \from \cN(c)$.
Then $\Dec(\sk, \tilde{c}) = \bot$ only if every block $r_i \Vert y_i$ contains an error.
Each bit is correct with probability $\frac{1+\rho}{2}$, so each block contains an error with probability $1 - \left(\frac{1+\rho}{2}\right)^{2 \ell}$.
The probability that every block contains an error is therefore
\[
    \left(1 - \left(\frac{1+\rho}{2}\right)^{2 \ell}\right)^{\secpar^2} \le \left(1 - \frac{\secpar^{\rho^2/2}}{\secpar^2}\right)^{\secpar^2} = \negl
\]
after simplifying.

This construction works when $\rho = 1-\smallO{1}$, but it fundamentally breaks down if $\rho=1-\bigOmega{1}$.
The issue is that pseudorandomness requires our seeds $r_i$ to each have length at least $\ell = \smallOmega{\log n}$ in order to avoid collisions; whereas any completeness-soundness gap requires $\ell = \bigO{\log n}$, because a constant noise rate will produce an error on a string of length $\ell$ with probability $1 - \exp(-\ell)$.

We show that this threshold is not a result of an insufficiently clever scheme, but a fundamental barrier.
That is, while we have just seen that it is easy to build a PRC for any $\smallO{1}$ error rate in the random oracle model, it is impossible to build a PRC for any $\bigOmega{1}$ error rate with black-box reductions to standard generic assumptions.

\subsection{The proof strategy}
\label{subsec:techo-strategy}
The remainder of our overview is devoted to proving the impossibility.
The bulk of the proof is showing that queries to the random oracle are not useful for constructing a PRC.

In order to show that a random oracle is not useful in a two-party (e.g., encoder and decoder) protocol against a third-party eavesdropper (e.g., pseudorandomness distinguisher), a standard methodology in the black-box separation literature is to argue that the eavesdropper can learn all the \emph{intersection queries} \cite{STOC:ImpRud89,C:BarMah09}.
The intersection queries are those made by both parties engaging in the protocol; if the eavesdropper knows all the intersection queries, then the participating parties cannot establish a shared secret.

However, this is not possible in our setting.
Unlike every prior work we are aware of, we are in the \emph{secret-key} setting, where the encoder and decoder share $\sk$.
The encoder and decoder can simply both query, say, $\sk$, making $\sk$ an intersection query that the eavesdropper has no way of learning.
In fact, the encoder and decoder can prevent the eavesdropper from learning any query at all by prepending $\sk$ to all of their queries.
One can therefore think of our model as a {\em secret random oracle model}, where only the secret key holders have access to the oracle.
Our approach will have to demonstrate that the noise channel makes it impossible for the encoder and decoder to put this secret random oracle to good use.

Consider an arbitrary PRC for $\cN_{\rho}$, say $(\Enc^\cR, \Dec^\cR)$, that uses a random oracle $\cR : \{0,1\}^\secpar \to \{0,1\}$.
Instead of defining an adversary, we will compile this PRC into a new one, say $(\widehat{\Enc}^\cR, \widehat{\Dec}^\cR)$, where the decoder uses \emph{one fewer query} to $\cR$, at the cost of a somewhat longer secret key and a small loss in completeness.\footnote{In \Cref{sec:impossibility} we will remove all queries at once. We take the one-query-at-a-time approach in this overview because we find it conceptually somewhat simpler at this level of technicality.}
By iterating this transformation, we can eliminate all queries from the decoder to the oracle---resulting in a PRC that is easily seen to be impossible: If the decoder makes no queries, then the encoder can simulate the oracle locally, and the oracle can be removed altogether.
But any PRC is in particular a pseudorandom encryption scheme, which cannot have statistical security in the standard model.

So, how does our compiled scheme $(\widehat{\Enc}^\cR, \widehat{\Dec}^\cR)$ work?
Let $f_\sk$ be the function mapping the received string $x$ to the first query $\Dec^\cR(\sk, x)$ makes to $\cR$.
Suppose for simplicity that $f_\sk$ is deterministic.
The basic idea of our compiled scheme is to transfer the oracle values on all queries $q$ such that $\Pr_{x \from \{0,1\}^n}[f_\sk(x) = q] \ge \tau$ to the secret key, for some threshold $\tau = 1/\poly$.
We call such queries ``low entropy'' and otherwise ``high entropy.''
Since there can be at most $1/\tau$ low entropy queries, our new secret key $\widehat{\sk}$ is at most $\bigO{\secpar/\tau}$ bits longer than $\sk$.

The compiled scheme will simulate $(\Enc^\cR, \Dec^\cR)$, responding to oracle queries as follows:
\begin{itemize}
    \item For low entropy queries, both $\widehat{\Enc}$ and $\widehat{\Dec}$ consult $\widehat{\sk}$.
    \item If the first query made by $\Dec$ is a high entropy query (that is, it is not in $\widehat{\sk}$), then $\widehat{\Dec}$ samples a fresh random response.
    \item For all other queries, both $\widehat{\Enc}$ and $\widehat{\Dec}$ respond according to $\cR$.
\end{itemize}
Observe that the first oracle query made by $\Dec$ is removed in $\widehat{\Dec}$, so our new decoder makes one fewer query.
Pseudorandomness and soundness for the compiled scheme are immediate, so we must only consider completeness.
Recall the completeness experiment:
\begin{enumerate}
    \item Sample $\widehat{\sk} \from \{0,1\}^\secpar$ uniformly at random.
    \item Generate a codeword $c \from \widehat{\Enc}^\cR(\widehat{\sk})$.
    \item Add noise to the codeword, $\tilde{c} \from \cN_{\rho}(c)$.
    \item Apply the decoder, $\widehat{\Dec}^\cR(\widehat{\sk}, \tilde{c})$. If the result is $\bot$, the experiment fails; if it is 1 then the experiment succeeds.
\end{enumerate}
The only case where $(\widehat{\Enc}^\cR, \widehat{\Dec}^\cR)$ might fail to perfectly simulate $(\Enc^\cR, \Dec^\cR)$ is when the first query made by $\Dec$, say $q$, is a high entropy query: If it happens that $\widehat{\Enc}$ had also queried $q$ in Step 2, then $\widehat{\Enc}$ and $\widehat{\Dec}$ will use different values for $\cR(q)$---potentially ruining completeness.
We argue in \Cref{subsec:techo-remove-one-query} that the error channel $\cN_{\rho}$ prevents $\widehat{\Enc}$ and $\widehat{\Dec}$ from making the same high entropy query.

\subsection{Bounding the error of the compiler}
\label{subsec:techo-remove-one-query}

Recall the function $f_\sk$, which maps the string $\tilde{c}$ received in Step 4 to the first oracle query $\widehat{\Dec}$ makes.
For this overview we assume $f_\sk$ is deterministic to simplify notation and terminology.
If we let $S$ be the set of high entropy queries\footnote{
    Remember that we define high entropy queries with respect to the decoder (specifically $f_\sk$). If the decoder, run on a random input, is unlikely to make $q$ as its first query, then $q$ is a ``high entropy query'' even if the encoder queries $q$ with probability 1.}
made by $\widehat{\Enc}^\cR(\widehat{\sk})$ in step 2, then our aim is to show that the following probability is $\negl$:
\begin{equation} \label{eq:compiler-1}
    \Pr_{\substack{c \from \widehat{\Enc}^\cR(\widehat{\sk}) \\ \tilde{c} \from \cN_{\rho}(c)}}[f_\sk(\tilde{c}) \in S].
\end{equation}
Again, if $f_\sk(\tilde{c}) \not\in S$ then our compiled scheme perfectly simulates the original one, so this will complete our proof.

The first simplifying observation is that, while $c$ has an unknown distribution for a given $(\widehat{\sk}, \cR)$, the \emph{marginal} distribution on just $(\widehat{\sk}, c)$ is actually uniform.
The reason is simple: If $c \from \widehat{\Enc}^\cR(\widehat{\sk})$ was non-uniform conditioned on $\widehat{\sk}$ alone, then the decoder could (inefficiently) distinguish codewords from random strings without using $\cR$ at all!
But, as we mentioned earlier, statistically ``pseudorandom'' encryption is not possible in the standard model (even with an inefficient decoder).
For a formal proof with quantitative bounds on the distance to uniformity, see \Cref{lem:stats}.

So we have the following equivalence with \Cref{eq:compiler-1}, where the codeword is now sampled \emph{uniformly at random} instead of using the encoder:
\[
    \Pr_{\substack{c \from \widehat{\Enc}^\cR(\widehat{\sk}) \\ \tilde{c} \from \cN_{\rho}(c)}}[f_\sk(\tilde{c}) \in S] = \Pr_{\substack{x \from \{0,1\}^n \\ \tilde{x} \from \cN_\rho(x)}}[f_\sk(\tilde{x}) \in S].
\]
The set $S$ of high entropy queries made by the encoder clearly does not depend on the error introduced by $\cN_\rho$.
So by a union bound, if we let $Q$ be an upper bound on the number of queries made by the decoder (so $\abs{S} \le Q$),
\begin{align*}
    \Pr_{\substack{x \from \{0,1\}^n \\ \tilde{x} \from \cN_\rho(x)}}[f_\sk(\tilde{x}) \in S] &\le Q \cdot \max_{f,g \in \cH_\tau} \Pr_{\substack{x \from \{0,1\}^n \\ \tilde{x} \from \cN_\rho(x)}}[f(\tilde{x}) = g(x)],
\end{align*}
where $\cH_\tau = \{f,g : \abs{f^{-1}(g(x^*))} / 2^n \le \tau \text{ for all } x^* \in \{0,1\}^n\}$, capturing the fact that every $q \in S$ is at most $\tau$-likely to be output by $f$.

Now it suffices to show the following simple and completely self-contained mathematical statement:
For every $\rho = 1-\Omega(1)$ and every pair of functions $f, g \in \cH_\tau$,
\begin{equation} \label{eq:compiler-2}
    \Pr_{\substack{x \from \{0,1\}^n \\ \tilde{x} \from \cN_\rho(x)}}[f(\tilde{x}) = g(x)] = \tau^{\Omega(1)}.
\end{equation}
Since $\tau$ is an arbitrary inverse-polynomial, this will imply that \Cref{eq:compiler-1} is $\negl$.

We will formally prove \Cref{eq:compiler-2} using hypercontractivity later, in \Cref{lem:main}.
For now, let us just try to give some basic intuition, aided by \Cref{fig:partition}.
What we need to show is that if two strings $x$ and $\tilde{x}$ are each random on their own and only $\rho$-correlated together, then very little information can be consistently extracted from them.\footnote{Note that it is essential that $x$ and $\tilde{x}$ are each random on their own, as otherwise error correction would make it possible to consistently extract lots of information!}
To understand why this is true, consider the special case that $f = g$.
We can view $f$ as a partition of $\{0,1\}^n$, with a label on each cell corresponding to the output of $f$ on any string in the cell.
Since $f$ is high entropy (i.e., $(f, f) \in \cH_\tau$), each cell in the partition is small.
The key point is that in the hypercube $\{0,1\}^n$, almost all of the probability mass is on the cell boundaries---therefore $f(x) \ne f(\tilde{x})$ with high probability.
Proving that small subsets of the hypercube have mass concentrated on their boundaries is highly non-trivial, and is encapsulated in the hypercontractivity theorem.

% It is very similar to the fact that the noisy hypercube is a small-set expander, which roughly means that in the high-dimensional hypercube, the mass of any small set concentrates on the boundary.
% For this overview, we will satisfy ourselves with the simpler proof that the hypercube is a small-set expander, adapted from \cite{KKL88}.
% \sam{TODO add SSE proof}

\begin{figure}
    \centering
    \begin{tikzpicture}[scale=6]
        \draw[thick] (0,0) rectangle (1,1);
        
        \draw (0.8,0) -- (0.3,1);
        \draw (0.40,0) -- (1,0.60);
        \draw (0,0.70) -- (1,0.40);
        \draw (0.15,0) -- (0.35,1);
        \draw (0,0.50) -- (1,0.80);
        
        % Point and ball
        \filldraw (0.24,0.55) circle (0.01);
        \draw[dashed] (0.24,0.55) circle (0.1);

        % Labels
        \node at (0.24,0.55) [below left] {$x$};
        \node at (0.3,0.45) [right] {$\cN_{\rho}(x)$};
    \end{tikzpicture}
    \caption{The space of strings $x \in \{0,1\}^n$ received by the decoder can be partitioned according to the first oracle query $f(x)$ the decoder makes upon receiving $x$. \Cref{eq:compiler-2} says that in high dimensions, if every cell in the partition is small, then a small Hamming ball centered at a random string has its mass distributed across many cells---i.e., no function $g(x)$ can guess where $\tilde{x} \from \cN_\rho(x)$ will land. Thus the noise destroys the decoder's ability to make high-entropy intersection queries.}
    \label{fig:partition}
\end{figure}
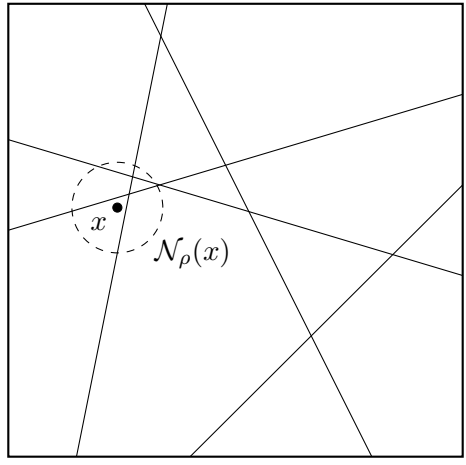

\subsection{Upgrading to a crypto oracle}
\label{subsec:techo-crypto-oracle}

It turns out that our random oracle impossibility easily extends to our full result by using the ideas of \cite{LMW25}.
That work introduces \emph{crypto oracles} and shows that they are capable of implementing most generic cryptography.
A crypto oracle is an efficient, stateless algorithm with access to a secret random oracle.

Because a PRC is a ``secret-key'' primitive---that is, the pseudorandomness adversary cannot access some secret randomness---our encoder and decoder can simply use the (public) random oracle to simulate the crypto oracle.
That is, by prepending all of their queries to the random oracle with the secret key, the encoder and decoder can effectively share a secret random oracle.

Now any construction of a PRC relative to a crypto oracle $\cB^\cR$ immediately implies a PRC in the secret random oracle model, where one simply lets the encoder and decoder simulate $\cB^\cR$ using access to the secret oracle $\cR$.\footnote{
    In a crypto oracle model, a PRC guarantees that an adversary cannot break pseudorandomness given access to $\cB^\cR$. In the compiled scheme, the adversary does not even have access to $\cB^\cR$ and, hence, can only become weaker and achieve a lower advantage.}
Consequently, our theorem immediately implies our separation result in the crypto oracle model.

\section{Toolkit}
\label{sec:toolkit}

\subsection{Notation}
We use $\secpar$ for the security parameter and $\negl$ for a negligible function, i.e., for any polynomial $f(\secpar)$, it holds that $\negl<1/f(\secpar)$ for all large enough $\secpar$.
We use $\N$ to denote the set of positive integers.
In this work ``$\log$'' will denote the base-2 logarithm.
We assume that the reader is familiar with the basic ideas of information theory, such as Shannon entropy, statistical distance, and KL divergence.
For random variables $X$ and $Y$, we write $\entropy{X}$ for the binary entropy of $X$; $\entropy{X \vert Y}$ for the binary entropy of $X$ conditioned on $Y$; $\SD{X,Y}$ for the statistical distance (i.e. total variation distance) between $X$ and $Y$; and $\KL{X}{Y}$ for the binary KL divergence of $X$ from $Y$.

For sets $\cX, \cY$, we will write $f : \cX \to \cY$ and say that $f$ is a ``randomized function'' if, for each $x \in \cX$, $f(x)$ is a random variable on $\cY$.
We use this terminology for convenience of notation; we find that the more standard approach of having $f$ map $\cX$ to distributions on $\cY$ would be cumbersome, and calling $f$ a ``randomized algorithm'' is not appropriate in all of our instances.

\subsection{Information theory}
In this section we prove \Cref{lem:stats}, which we will use twice in our main proof: Once in order to invoke our key technical lemma, \Cref{lem:main}; and once in order to see that statistical security is impossible in the plain model.
First we relate the Shannon entropy to the statistical distance from the uniform distribution.

\begin{fact} \label{fact:pinsker-app}
    Let $X$ be any random variable taking values in a finite set $\cS$, and let $R$ be a uniformly random sample from $\cS$.
    If $\varepsilon = \SD{X,R} \le 1/4$, then
    \[
        2\varepsilon^2 \le \log\abs{S} - \entropy{X} \le 2 \varepsilon \log\abs{S} + 2 \sqrt{\varepsilon}.
    \]
\end{fact}
\begin{proof}
    For the inequality on the left, first observe that $\KL{X}{R} = \log\abs{S} - \entropy{X}$.
    Then by Pinsker's inequality, $2 \varepsilon^2 \le \KL{X}{R} \cdot \ln 2 \le \KL{X}{R}$.
    
    For the inequality on the right, we have $\log\abs{S} - H(X) \le 2 \varepsilon \log\abs{S} + 2 \varepsilon \log(1/2\varepsilon)$ from \cite[Theorem 16.3.2]{cover-thomas}.
    We simplify the second term as $2 \varepsilon \log(1/2\varepsilon) \le 2 \varepsilon \log(1/\varepsilon) \le 2 \sqrt{\varepsilon}$ for $\varepsilon \le 1/4$.
\end{proof}

\noindent
We now deduce \Cref{lem:stats}.
Roughly, \Cref{lem:stats} states that a finite-length key can only be used to statistically hide a finite number of messages.

\begin{lemma}\label{lem:stats}
    Let $\{\cD_{\key}\}_{\key \in \{0,1\}^\ell}$ be a family of distributions over $\{0,1\}^n$.
    Let $\sk$ be a uniformly random sample from $\{0,1\}^\ell$ and let $x, x_1, \dots, x_m$ be independent samples from $\cD_{\sk}$.

    If $(x_1, \dots, x_m)$ are jointly $\varepsilon$-close to uniform in statistical distance, then $(\sk, x)$ is
    \[
        \bigO{\sqrt{\varepsilon n + \sqrt{\varepsilon}/m + \ell/m}}
    \]
    close to uniform in statistical distance.
\end{lemma}
%\mingyuan{$m$ is a polynomial? For the tight bound, we will need to set $m$ to be super-polynomial. So, will need to change $\eps$-close to $(m\eps)$-close.}
\begin{proof}
    Let $\sk$ be a random string from $\{0,1\}^\ell$, and let $x, x_1, \dots, x_m$ be samples from $\cD_{\sk}$.
    We have
    \begin{align*}
        H(x_1, \dots, x_m) &\le H(x_1, \dots, x_m, \sk) \\
        &= H(\sk) + H(x_1, \dots, x_m | \sk) \\
        &= \ell + m \cdot H(x | \sk),
    \end{align*}
    which means that $H(x | \sk) \ge (H(x_1, \dots, x_m) - \ell) / m$.
    
    Let $R$ be a uniformly random sample from $\{0,1\}^{\ell + n}$.
    Fact~\ref{fact:pinsker-app} implies the following two inequalities:
    \begin{itemize}
        \item[$\bullet$] $\SD{(\sk, x), R} \le \sqrt{n + \ell - H(\sk, x)} = \sqrt{n - H(x | \sk)}$, and
        \item[$\bullet$] $H(x_1, \dots, x_m) \ge nm - 2 \varepsilon nm - 2 \sqrt{\varepsilon}$.
    \end{itemize}
    Therefore
    \begin{align*}
        \SD{(\sk, x), R} &\le \sqrt{n - H(x | \sk)} \\
        &\le \sqrt{n - (H(x_1, \dots, x_m) - \ell) / m} \\
        &\le \sqrt{n - (nm - 2 \varepsilon nm - 2 \sqrt{\varepsilon} - \ell) / m} \\
        &= \sqrt{2 \varepsilon n + 2 \sqrt{\varepsilon}/m + \ell/m},
    \end{align*}
    completing the proof.
\end{proof}

\subsection{Pseudorandom codes}
\label{sec:prc-def}

Following~\cite{C:ChrGun24}, we define private-key pseudorandom codes as follows.
Our impossibility result will apply to the private-key setting, which immediately implies the corresponding impossibility for the public-key setting.
Therefore we omit definitions of public-key pseudorandom codes.%~\cite{CiC:FGJMMW25}.% I think it would be extremely confusing to cite this here: this paper is not about PRCs at all. And the definition of ``public-key PRC'' you might infer from it is totally different from anything we consider here. In particular, everything in that paper can be done from OWFs (recall how that scheme is either not undetectable or not robust, depending how it's implemented).
% We can totally cite it elsewhere, but it should be somewhere we're discussing watermarking, and it should be clear that we're talking about ``unforgeability''---for the reasons I explained and we wrote in the SoK paper, I disagree with calling them ``public-key watermarks''---and it should be clear that it's not robust and undetectable, so it's not really about PRCs.

\begin{definition}[Private-key PRC]\label{def:prc} A {\em private-key $(\delta, \varepsilon, \mu)$ pseudorandom error-correcting code} over an alphabet $\Sigma$ and a channel $\cE\colon\Sigma^*\rightarrow\Sigma^*$ consists of a tuple of PPT algorithms $(\KeyGen, \Enc, \Dec)$:
\begin{itemize}
    \item A key generation algorithm $\KeyGen$ that samples a secret key $\sk\in\zo^{\ell(\secpar)}$. 
    \item An encoding algorithm $\Enc$ that takes the secret key $\sk\in\zo^{\ell(\secpar)}$ and the message $m\in\Sigma^{k(\secpar)}$ as input and outputs a codeword $c\in\Sigma^{n(\secpar)}$.
    \item A decoding algorithm that takes the secret key $\sk\in\zo^{\ell(\secpar)}$ and a (potentially erroneous) codeword $c\in\Sigma^{n(\secpar)}$ as input and outputs a message $m\in\Sigma^{k(\secpar)}\cup\{\bot\}$.
\end{itemize}
These algorithms must satisfy the following properties:
\begin{itemize}
    
    \item {\bfseries $\delta$-Completeness (robustness).} For any message $m\in\Sigma^k$, it holds that
    $$\prob{\left.\Dec(\sk,\widetilde{c}) = m \ \middle|\  \begin{gathered}
        \sk\from\KeyGen(\cdot)\\
        c\from\Enc(\sk,m)\\
        \widetilde{c}\from \cE(c)
    \end{gathered}\right.} \geq 1-\delta.$$
    
    \item {\bfseries $\eps$-Pseudorandomness.} For any PPT adversary $\cA$, it holds that
    $$
    \abs{\prob{\left.\cA^{\Enc(\sk,\cdot)}(\cdot) = 1 \ \middle|\   \sk \from \KeyGen(\cdot)\right.} 
    - \prob{\cA^{\cU(\cdot)}(\cdot) = 1}} \leq \eps,
    $$
    where $\Enc(\sk,\cdot)$ generates fresh encodings even if the same message is queried twice and $\cU(\cdot)$ is an oracle that simply outputs fresh random strings.

    \item {\bfseries $\mu$-Soundness.} For a fixed codeword $c^*\in\Sigma^n$, it holds that
    $$\prob{\Dec(\sk,c^*) = \bot \ \vert\ \sk\from\KeyGen(\cdot)}\geq 1- \mu .$$
\end{itemize}
We will drop the specification of $\delta, \varepsilon, \mu$ in the case that they are all bounded by $\negl$.
\end{definition}

\paragraph{Zero-bit PRC.} We consider PRCs with only one possible message, $m=1$, which are known as {\em zero-bit} PRCs and are useful for watermarking.
For random errors, one can convert a zero-bit PRC to a many-bit PRC simply by concatenating codewords and random strings \cite{C:ChrGun24}.
Conversely, a zero-bit PRC is immediately implied by any PRC with $k > 0$, so our impossibility result is not weakened at all by restricting to this case.
We therefore drop the message in the remainder of this paper, writing $\Enc(\sk)$ instead of $\Enc(\sk, 1)$.
The decoder $\Dec(\sk, x)$ outputs $1$ or $\bot$, indicating success or failure, respectively.

\paragraph{Binary Alphabet.} In this work we will restrict ourselves to a binary alphabet, i.e. $\Sigma = \zo$.
Note that a pseudorandom code over large alphabet implies a pseudorandom code for binary alphabet.
Since we are proving an impossibility result, this restriction only makes our lower bound stronger.

\paragraph{Noise channel.} In this work we will only consider the binary symmetric channel, which introduces i.i.d bit-flip errors on the codeword bits.
Adversarial noise channels such as those considered in \cite{AACDG25} can easily simulate the binary symmetric channel, so our impossibility result applies to pseudorandom codes for those channels as well.
We will denote the binary symmetric channel by $\cN_{\rho}$, defined as
\[
    \cN_{\rho}(x) = x \oplus e
\]
where $e \from \Ber((1-\rho)/2)^n$.
We use $\Ber(p)$ to denote the Bernoulli distribution with probability $p$, and $\oplus$ to denote the bitwise XOR.
In other words, each symbol of $x$ is replaced by a random element with probability $1-\rho$ in $\cN_{\rho}(x)$.

\subsection{Boolean analysis and hypercontractivity}

A core part of our argument makes use of Boolean analysis.
We introduce minimal notation here to facilitate our proof; we refer the reader to \cite{ryan-book} for more details.
For any Boolean function $f : \{0,1\}^n \to \R$, its $\ell$-norm $\Vert f\Vert_\ell$ is defined as
$$ \Vert f\Vert_\ell = \left(\E_x[f(x)^\ell]\right)^{1/\ell}.$$
Unless otherwise specified, expectations are assumed to be taken over the uniform distribution.
For any two Boolean functions $f$ and $g$, their {\em inner product} is defined as
$$\inner{f}{g} = \E_x[f(x)\cdot g(x)].$$
By Cauchy-Schwartz, we have
$$\inner f g\leq \Vert f\Vert_2\cdot \Vert g\Vert_2.$$

We now define the {\em noise operator $T_\rho$}.
For any Boolean function $f$, the noise operator $T_\rho$ on the function $f$ defines a new function $ T_\rho f $ as
$$(T_\rho f)(x) = \E_{y\draw\cN_\rho (x)}[f(y)].$$
The following is a special case of the {\em hypercontractivity theorem} \cite{FOCS:KahKalLin88}.

\begin{theorem}[Hypercontractivity, \cite{FOCS:KahKalLin88}] \label{theorem:hypercontractivity}
    For any function $f : \{0,1\}^n \to \R$ and $\rho \in [0,1]$,
    \[
        \norm{T_{\rho} f}_2 \le \norm{f}_{1+\rho^2}.
    \]
\end{theorem}
%\TODO{include a proof or cite}

Hypercontractivity is the core ingredient in our main technical lemma, \Cref{lem:main}.

\begin{lemma}\label{lem:main}
    Let $\alpha > 0$. Suppose $f, g : \{0,1\}^n \to \calY$ are (randomized) functions such that, for all $x^* \in \{0,1\}^n$, $$\prob{f(x) = g(x^*)\ \vert\ x\draw\zo^n}\leq \alpha.$$ Then, for any $\rho \in [0,1]$, it holds that
    \[
        \prob{\left. f(\widetilde{x}) = g(x) \ \middle\vert\ \begin{gathered}
            x \draw \{0,1\}^n \\ \widetilde{x} \draw\calN_{\rho}(x)
        \end{gathered}\right.} \leq \alpha^{\frac{1}{2} \cdot \left(\frac{1-\rho^2}{1+\rho^2}\right)}.
    \]
\end{lemma}
\begin{proof}
The proof works by a careful application of Cauchy-Schwarz and hypercontractivity.
For $y \in \cY$ and $x \in \{0,1\}^n$, let $p_y(x) = \Pr[f(x) = y]$ and $q_y(x) = \Pr[g(x) = y]$.
Let $\cY^* = \{y \in \cY : \norm{q_y}_1 > 0\}$.
\begin{align*}
    &\prob{\left. f(\widetilde{x}) = g(x) \ \middle\vert\ \begin{gathered}
            x \draw \{0,1\}^n \\ \widetilde{x} \draw\calN_{\rho}(x)
        \end{gathered}\right.} \\
    &=
    \E_{x \from \{0,1\}^n} \sum_{y \in \cY^*} q_y(x) \cdot (T_{\rho} p_y)(x) \\
    &= \sum_{y \in \cY^*} \inner{q_y}{T_{\rho} p_y} \\
    &\le \sum_{y \in \cY^*} \norm{q_y}_2 \cdot \norm{T_{\rho}p_y}_2 & \text{(Cauchy-Schwarz)} \\
    &\le \sum_{y \in \cY^*} \norm{q_y}_2 \cdot \norm{p_y}_{1+\rho^2} & \text{(Hypercontractivity)} \\
    &\le \sqrt{\sum_{y \in \cY^*} \norm{q_y}_2^2} \cdot \sqrt{\sum_{y \in \cY^*} \norm{p_y}_{1+\rho^2}^2} & \text{(Cauchy-Schwarz)} \\
    &= \sqrt{\sum_{y \in \cY^*} \E_x q_y(x)^2} \cdot \sqrt{\sum_{y \in \cY^*} \left(\E_x p_y(x)^{1+\rho^2}\right)^{2/(1+\rho^2)}} \\
    &\le \sqrt{\sum_{y \in \cY^*} \E_x q_y(x)} \cdot \sqrt{\sum_{y \in \cY^*} \left(\E_x p_y(x)\right)^{2/(1+\rho^2)}} \\
    \intertext{Letting $P(y) = \Pr_{x \from \{0,1\}^n}[f(x) = y]$ and $Q(y) = \Pr_{x \from \{0,1\}^n}[g(x) = y]$,}
    &= \sqrt{\sum_{y \in \cY^*} Q(y)} \cdot \sqrt{\sum_{y \in \cY^*} P(y)^{2 / (1+\rho^2)}} \\
    &= 1 \cdot \sqrt{\sum_{y \in \cY^*} P(y) \cdot P(y)^{2 / (1+\rho^2)-1}} \\
    &\le \sqrt{\sum_{y \in \cY^*} P(y) \cdot \alpha^{2 / (1+\rho^2)-1}} & \text{(Assumption on $f$)} \\
    &= \alpha^{\frac{1}{2} \left(\frac{1-\rho^2}{1+\rho^2}\right)}. & \qedhere
\end{align*}
% Now, for {\em randomized} $f$ and $g$, define
% $$\alpha_{s,r}:=\prob{f(x;s) = g(x^*;r)\ \vert\ x\draw\zo^n}.$$
% Note that $\expsub{s,r}{\alpha_{s,r}}\leq \alpha$.
% Now, we prove this lemma by
% \begin{align*}
%         &\prob{\left. f(\widetilde{x}) = g(x) \ \middle\vert\ \begin{gathered}
%             x \draw \{0,1\}^n \\ \widetilde{x} \draw\calN_{\rho}(x)
%         \end{gathered}\right.} 
%     \\
% &=\expsub{s,r}{
%         \prob{\left. f(\widetilde{x};s) = g(x;r) \ \middle\vert\ \begin{gathered}
%             x \draw \{0,1\}^n \\ \widetilde{x} \draw\calN_{\rho}(x)
%         \end{gathered}\right.} 
%     }\\
% &\leq \expsub{s,r}{\left(\alpha_{s,r}\right)^{\frac12\cdot\left(\frac{1-\rho^2}{1+\rho^2}\right)}}\tag{Deterministic functions}\\
% &\leq \left(\expsub{s,r}{\alpha_{s,r}}\right)^{\frac12\cdot\left(\frac{1-\rho^2}{1+\rho^2}\right)}\tag{Jensen's inequality}\\
% & \leq \alpha^{\frac12\cdot\left(\frac{1-\rho^2}{1+\rho^2}\right)}
% \end{align*}
\end{proof}

\subsection{Crypto oracles}
\label{sec:cryptooracle}
The recent work of \cite{LMW25} introduced ``crypto oracles'' for the purposes of an impossibility result for doubly-efficient private information retrieval.

\begin{definition}[Crypto oracle, \cite{LMW25}]
    A \emph{crypto oracle} is a function $\cB^\cR$ where $\cB$ is a stateless, polynomial time, deterministic Turing machine with oracle access to a secret random function $\cR : \{0,1\}^* \to \{0,1\}$.
\end{definition}

The same work proved the following, which we have heavily paraphrased for simplicity.
See their paper for details.

\begin{theorem}[\cite{LMW25}]
    There exist crypto oracles that implement:
    \begin{itemize}
        \item VBB obfuscation for Turing machines, and
        \item the generic multilinear group.
    \end{itemize}
\end{theorem}

\section{The impossibility}
\label{sec:impossibility}

\subsection{Compiling out the random oracle}

In this section, we present a compiler that compiles any PRC in the random oracle model into a PRC in the information-theoretic setting.
The new scheme has the same soundness and pseudorandomness, but it requires a larger secret key and incurs a loss in the completeness error.

We assume $\KeyGen^\cR$ does not make any oracle queries.
This is without loss of generality because we can simply include any $\cR$ queries/responses used by $\KeyGen$ in the secret key.

In the following lemma and proof, we leave the dependence of functions of the security parameter on the security parameter implicit, writing e.g. $\tau$ instead of $\tau(\secpar)$.

\begin{lemma}\label{lem:ro-1}
Let $\ell, Q : \N \to \N$ be functions of the security parameter.
Let $(\KeyGen, \Enc^\cR, \Dec^\cR)$ be any zero-bit $(\delta,\eps,\mu)$-PRC construction in the random oracle model where $\Enc^\cR$ and $\Dec^\cR$ each make at most $Q$ queries to $\cR$ and $\KeyGen(\secparam)$ outputs keys of length at most $\ell$.

Then for any $\tau : \N \to (0,1)$, there exists a $(\delta',\eps',\mu')$-PRC construction $(\KeyGen', \Enc', \Dec')$ in the standard model where
\begin{itemize}
    \item $\delta' = \delta + 2^{-\secpar} \cdot Q / \tau + Q^2\cdot\tau^{\frac{1}{2} \cdot \left(\frac{1-\rho^2}{1+\rho^2} \right)} + \bigO{Q^2\cdot\sqrt{\varepsilon n}}$,
    \item $\eps' = \eps$,
    \item $\mu' = \mu$, and
    \item the size of the new secret key is $\ell + \bigO{Q \cdot \secpar^2/\tau}$.
\end{itemize}
\end{lemma}

\begin{figure}[!ht]
\begin{mdframed}
Let $\tau\in(0,1)$ be any number. Let $(\KeyGen,\Enc^\cR,\Dec^\cR)$ be a PRC in the random oracle model where $\Dec^\cR$ makes at most $Q = Q(\secpar)$ queries to $\cR$. We compile it into a new PRC scheme $(\KeyGen',\Enc',\Dec')$ in the standard model as follows.

\begin{itemize}
    \item $\KeyGen'(\secparam)$:
    \begin{enumerate}
        \item Sample $\sk \draw \KeyGen(\secparam)$ and a random function $F_0: \{0,1\}^\secpar \to \{0,1\}$.
        \item Initialize $S = \emptyset$ and repeat the following for $\lceil \lambda/\tau \rceil$ times:
        \begin{itemize}
            \item Sample $x \draw \{0,1\}^n$ and run $\Dec^{F_0}(\sk, x)$, recording each query-response pair in $S$.
        \end{itemize}
        \item Output the new secret key as $\sk' = (\sk, S)$.
    \end{enumerate}

    \item $\Enc'(\sk')$: Parse $\sk' = (\sk, S)$ and sample a random function $F_1 : \{0,1\}^\secpar \to \{0,1\}$ that is consistent with $S$. Run $\Enc^{F_1}(\sk)$ and output the result.
    
    \item $\Dec'(\sk',x)$: Parse $\sk' = (\sk, S)$ and sample a random function $F_2 : \{0,1\}^\secpar \to \{0,1\}$ that is consistent with $S$. Run $\Dec^{F_2}(\sk, x)$ and output the result.
\end{itemize}

\end{mdframed}
\caption{Our compiler}
\label{fig:compiler}
\end{figure}

\begin{proof}[Proof of Lemma~\ref{lem:ro-1}]

Our compiler is presented in Figure~\ref{fig:compiler}. We now proceed to prove its properties as stated in Lemma~\ref{lem:ro-1}.

\paragraph{Secret key size.} The size of the secret key grows by the size of $S$, which is $\bigO{Q \cdot \secpar^2/\tau}$ by construction.

\paragraph{Pseudorandomness \& soundness.} Observe that for the pseudorandomness (resp., soundness) property, the experiment only involves $\KeyGen'$ and the encoder $\Enc'$ (resp., the decoder $\Dec'$).
The distribution of any process involving only the encoder (resp., decoder) is identical in the compiled scheme to the original scheme.
This is because \emph{locally}, the encoder (resp., decoder) perfectly simulates the random oracle.
Consequently, the pseudorandomness and soundness guarantees do not change at all.

\paragraph{Completeness.}
If $\Enc^{F_1}$ and $\Dec^{F_2}$ do not make the same query (except for those contained in $S$), then the compiled scheme perfectly simulates the original scheme.
Therefore, it suffices to bound the probability that $\Enc^{F_1}$ and $\Dec^{F_2}$ both query their oracles on the same point $q$ which is not in $S$.
We refer to such a query as an {\em intersection query}.%
\footnote{Note that queries in $S$ are not counted as intersection queries.}

We will say that a query $q$ is \emph{$\tau$-heavy} for $\sk, F$ if
\[
    \Pr[\Dec^F(\sk, x) \text{ queries } q  \ \vert\ x \draw \{0,1\}^n] \geq \tau.
\]
Consider the PRC completeness experiment for our compiled scheme.

\begin{enumerate}
    \item[] \textbf{Completeness experiment:}
    \item Sample $\sk \draw \KeyGen(\secparam)$ and a random function $F_0: \{0,1\}^\secpar \to \{0,1\}$.
    \item Initialize $S = \emptyset$ and repeat the following for $\lceil\secpar/\tau\rceil$
    times:
    \begin{itemize}
        \item Sample $x \draw \{0,1\}^n$ and run $\Dec^{F_0}(\sk, x)$. Record each query-response pair in $S$.
    \end{itemize}
    \item Let $\sk' = (\sk, S)$. Sample random functions $F_1, F_2 : \{0,1\}^\secpar \to \{0,1\}$ consistent with $S$.
    \item Compute $c \draw \Enc^{F_1}(\sk)$ and $\tilde{c} \draw \cN_{\rho}(c)$.
    \item Compute $\Dec^{F_2}(\sk, \tilde{c})$.
\end{enumerate}
In this experiment we define two events:
\begin{itemize}
    \item $\Bad_1$: This event happens when $S$ does not contain all queries which are $\tau$-heavy for $\sk, F_2$.
    \item $\Bad_2$: This event happens when (1) $S$ contains all queries which are $\tau$-heavy for $\sk, F_2$, but (2) $\Enc^{F_1}(\sk)$ and $\Dec^{F_2}(\sk, \tilde{c})$ still make an intersection query (which is not included in $S$).
\end{itemize}
We conclude the completeness guarantee by proving the following two bounds:
\begin{enumerate}
    \item $\prob{\Bad_1}\leq 2^{-\secpar} \cdot Q / \tau$. This bound will hold for every fixed choice of $\sk$.
    \item $\prob{\Bad_2}\leq Q^2\cdot \left( \tau^{\frac{1}{2} \cdot \left(\frac{1-\rho^2}{1+\rho^2} \right)} + \bigO{\sqrt{\varepsilon n}} \right)$. This bound will only hold on average over $\sk$.
\end{enumerate}

\paragraph{Bounding $\Bad_1$.}
% Fix $\sk$ and $F_0$ after sampling them in step 1 of the PRC completeness experiment, and assume that $F_2 = F_0$.
% Observe that, since $F_0$ is only queried on points on which it is guaranteed to be consistent with $F_2$, the probability of $\Bad_1$ does not change if we set $F_2 = F_0$ (instead of sampling it as a random function consistent with $S$).
In this part we will consider $\sk$ as fixed.
For any $q \in \{0,1\}^\secpar$ that is $\tau$-heavy for $\sk, F_0$, the probability that $q$ is queried in any single iteration of the loop in step 2 is at least $\tau$.
Therefore, the probability that $q$ is \emph{never} queried throughout the $\lceil \secpar / \tau \rceil$ iterations of the loop in step 2 is at most
\[
    (1-\tau)^{\lceil \secpar / \tau \rceil} \leq 2^{-\secpar}.
\]
The total number of $\tau$-heavy queries for $\sk, F_0$ is upper bounded by $Q/\tau$ because $\Dec^{F_0}(\sk, \cdot)$ makes at most $Q$ queries.
So by a union bound, the probability that $S$ does not contain all the queries which are $\tau$-heavy for $\sk, F_0$ is at most $2^{-\secpar} \cdot Q / \tau$.

Of course, we are interested in the probability that $S$ contains all the queries which are $\tau$-heavy for $\sk, F_2$ (not $\sk, F_0$).
But these two probabilities are the same because, conditioned on $S$, $F_0$ and $F_2$ are equal on every point in $S$ and both uniformly random on every other point.

More formally, note that for a given $\sk$ we could equivalently define $F_0$ as follows:
\begin{enumerate}
    \item Initialize $S = \emptyset$ and repeat the following for $\lceil\secpar/\tau\rceil$
    times:
    \begin{itemize}
        \item Sample $x \draw \{0,1\}^n$ and run $\Dec^{(\cdot)}(\sk, x)$, simulating the oracle on-the-fly. Record each query-response pair in $S$.
    \end{itemize}
    \item Let $\sk' = (\sk, S)$. Sample a random function $F_0 : \{0,1\}^\secpar \to \{0,1\}$ consistent with $S$.
\end{enumerate}
Since the distribution of $S, F_0$ does not change if we sample $F_0$ in this way instead, the probability that $S$ contains all the queries which are $\tau$-heavy for $\sk, F_0$ does not change.
However, under this method of sampling it is clear that $S, F_0$ is distributed identically to $S, F_2$, so
\[
    \Pr[\Bad_1] \le 2^{-\secpar} \cdot Q / \tau.
\]

\paragraph{Bounding $\Bad_2$.} Now we are interested in the probability that both the encoder and decoder ask a query that is not $\tau$-heavy for $\sk, F_2$.
Whereas we bounded $\Pr[\Bad_1]$ for every fixed choice of $\sk$, we will use the randomness of the entire experiment to bound $\Pr[\Bad_2]$.

For $i,j\in[Q]$, let $\Bad_2^{i,j}$ denote the event that (1) $S$ contains all the $\tau$-heavy queries for $\sk, F_2$, but (2) the $i$-th query made by $\Enc^{F_1}(\sk)$ and the $j$-th query made by $\Dec^{F_2}(\sk, \tilde{c})$ are both some $q$ that is not in $S$.
By a union bound,
\[
    \prob{\Bad_2} \le \sum_{i,j\in[Q]}\prob{\Bad_2^{i,j}}.
\]
We will use Lemma~\ref{lem:main} to show that, for every $i,j \in [Q]$,
\[
    \prob{\Bad_2^{i,j}}\leq \tau^{\frac{1}{2} \cdot \left(\frac{1-\rho^2}{1+\rho^2} \right)} + \bigO{\sqrt{\varepsilon n}}.
\]
Let $f_{\sk'}$ be the (randomized) function with $\sk' = (S, \sk)$ hardwired that maps $\tilde{c}$ to the $j$-th query made by $\Dec^{F_2}(\sk, \tilde{c})$.
Let $h_{\sk'}$ be the (randomized) function with $\sk'$ hardwired that maps $c$ to the $i$-th query made by $\Enc^{F_1}(\sk)$, conditioned on $c \draw \Enc^{F_1}(\sk)$; let $g_{\sk'}$ be identical to $h_{\sk'}$ except that $g_{\sk'}(c)$ outputs $\bot$ whenever $h_{\sk'}$ outputs any query that is in $S$.
Observe that
\[
    \prob{\Bad_2^{i,j}} = \prob{\left. \neg \Bad_1 \wedge f_{\sk'}(\tilde{c}) = g_{\sk'}(c) \ \middle\vert\ \begin{gathered}
            c \draw \Enc^{F_1}(\sk) \\ \tilde{c} \draw\calN_{\rho}(c)
        \end{gathered}\right.}.
\]
Applying \Cref{lem:stats}, we have that $\SD{(\sk', c \draw \Enc(\sk')), (\sk', x \draw \{0,1\}^n\})} = \bigO{\sqrt{\varepsilon n}}$, so
\begin{align*}
    & \prob{\left. \neg \Bad_1 \wedge f_{\sk'}(\tilde{c}) = g_{\sk'}(c) \ \middle\vert\ \begin{gathered}
            c \draw \Enc^{F_1}(\sk) \\ \tilde{c} \draw\calN_{\rho}(c)
        \end{gathered}\right.} \\
    &\le \prob{\left. \neg \Bad_1 \wedge f_{\sk'}(\widetilde{x}) = g_{\sk'}(x) \ \middle\vert\ \begin{gathered}
            x \draw \{0,1\}^n \\ \widetilde{x} \draw\calN_{\rho}(x)
        \end{gathered}\right.} + \bigO{\sqrt{\varepsilon n}}.
\end{align*}
Recall that $\neg \Bad_1$ means $S$ contains all the $\tau$-heavy queries for $S, F_2$, and $g_{\sk'}(x)$ never outputs any query that is in $S$.
Therefore, for any $x^* \in \{0,1\}^n$,
$$\prob{f_{\sk'}(x) = g_{\sk'}(x^*) \ \vert\ x\draw\zo^n \wedge \neg \Bad_1}\leq \tau.$$
Since conditioning on $\neg \Bad_1$ does not affect the distribution on $x \from \{0,1\}^n, \tilde{x} \from \calN_{\rho}(x)$, Lemma~\ref{lem:main} implies that
\begin{align*}
    & \prob{\left. \neg \Bad_1 \wedge f_{\sk'}(\tilde{x}) = g_{\sk'}(x) \ \middle\vert\ \begin{gathered}
        x \draw \{0,1\}^n \\ \tilde{x} \draw\calN_{\rho}(x)
    \end{gathered}\right.} \\
    &\le \prob{\left. f_{\sk'}(\tilde{x}) = g_{\sk'}(x) \ \middle\vert\ \begin{gathered}
        x \draw \{0,1\}^n \\ \tilde{x} \draw\calN_{\rho}(x) \\ \neg \Bad_1
    \end{gathered}\right.}  \\
    &\leq \tau^{\frac{1}{2} \cdot \left(\frac{1-\rho^2}{1+\rho^2}\right)},
\end{align*}
completing the proof.
\end{proof}

\subsection{The main theorem}

Let $\cR:\zo^\secpar\to\zo$ be a random oracle which is only given to the algorithms of the pseudorandom code.
That is, we do not assume that pseudorandomness holds against adversaries who are allowed to access $\cR$.
We are now prepared to prove the following theorem.

\begin{theorem}
\label{thm:ro-main}
Let $(\KeyGen^{\cR}, \Enc^{\cR}, \Dec^{\cR})$ be a zero-bit $(\delta,\eps,\mu)$ pseudorandom error-correcting code making queries to a secret random oracle $\cR$.
Suppose that the code is robust to the $\rho$-noise channel, and that the number of queries made by one execution of $\KeyGen^{\cR}$, $\Enc^{\cR}$, and $\Dec^{\cR}$ is at most $Q$.
Then
\[
    \delta + \mu \ge 1 - Q^2 \cdot \secpar^{-\omega(1-\rho)} - \bigO{Q^2 \sqrt{\varepsilon n}}.
\]
\end{theorem}
\begin{proof}
    We first remove any oracle queries made by $\KeyGen^{\cR}$ by sampling the responses at random and including them in the secret key.
    If $\Enc^{\cR}$ or $\Dec^{\cR}$ tries to make any query whose response is included in the secret key, it uses the stored response instead of consulting $\cR$.
    This transformation does not change the parameters of our PRC at all, except to increase the length of the secret key by a polynomial in $\secpar$.

    Next, we apply \Cref{lem:ro-1} using $\tau = \secpar^{-c}$, where $c > 0$ is any constant.
    This yields a statistically-pseudorandom $(\delta',\eps,\mu)$-PRC construction $(\KeyGen', \Enc', \Dec')$ in the standard model with
    \[
        \delta' = \delta + \negl + Q^2\cdot\secpar^{-\frac{c}{2} \cdot \left(\frac{1-\rho^2}{1+\rho^2} \right)} + \bigO{Q^2\cdot\sqrt{\varepsilon n}}
    \]
    and $\poly$-length secret keys.
    Now by \Cref{lem:stats}, PRC codewords are $\bigO{\sqrt{\varepsilon n}}$-indistinguishable from random strings \emph{even to a distinguisher provided with the secret key}.
    Therefore the completeness-soundness gap is $\bigO{\sqrt{\varepsilon n}}$, i.e., $\delta' + \mu \ge 1-\bigO{\sqrt{\varepsilon n}}$.
    Substituting the above formula for $\delta'$ and simplifying yields the result.
\end{proof}

\begin{corollary}
    There do not exist statistically-secure pseudorandom codes for any constant-rate noise channel, relative to any crypto oracle.
\end{corollary}
\begin{proof}
    Suppose that $(\KeyGen^{\cO}, \Enc^{\cO}, \Dec^{\cO})$ is a pseudorandom code for some $(1-\Omega(1))$-noise channel, relative to some crypto oracle $\cO$.
    Assume it is a zero-bit PRC, by always using the all-zero message and disregarding the decoder's output.

    Define $(\overline{\KeyGen}^{\cR}, \overline{\Enc}^{\cR}, \overline{\Dec}^{\cR})$ to behave identically to $(\KeyGen^{\cO}, \Enc^{\cO}, \Dec^{\cO})$, except that they simulate $\cO$ using oracle access to a random function $\cR$.
    If $\cR$ is secret---that is, the pseudorandomness adversary is not allowed to query $\cR$---then $(\overline{\KeyGen}^{\cR}, \overline{\Enc}^{\cR}, \overline{\Dec}^{\cR})$ is a pseudorandom code.
    We can therefore invoke \Cref{thm:ro-main} with $(\overline{\KeyGen}^{\cR}, \overline{\Enc}^{\cR}, \overline{\Dec}^{\cR})$ to conclude that
    \[
        \delta + \mu \ge 1 - \negl,
    \]
    which means that there is a negligible completeness-soundness gap.
\end{proof}

% In this section, we give a simple bound on the existence of pseudorandom codes in the pure information-theoretic setting. Note that the bound we prove in this section stands even if the construction is inefficient. Namely, even if $\ell$ is {\em super-polynomial} in Lemma~\ref{lem:stats}, the security/completeness will still be proportional to $1/\ell$. This implies that, for instance, if one wants to achieve exponentially low security and completeness error, the length of the secret key should also grow exponentially.

% Formally, we have the following lemma.

% \begin{corollary} \label{corollary:stats}
%     Let $\ell, n : \N \to \N$ be polynomially-bounded functions of the security parameter $\secpar$.
%     Let $\Enc$ be a randomized algorithm taking keys in $\{0,1\}^\ell$ to outputs in $\{0,1\}^n$. Suppose that, for any function $m = \poly$, the distribution on $m$ outputs from $\Enc(\sk)$ is $\varepsilon$-close to uniform in statistical distance.
    
%     Then $(\sk, c \draw \Enc(\sk))$ is $\cO(\sqrt{\varepsilon n})$-close to $(\sk, u \draw \{0,1\}^n)$ in statistical distance.
% \end{corollary}

\addcontentsline{toc}{section}{References}
\bibliographystyle{alphaurl}
\bibliography{abbrev0,crypto,local}

% \appendix

% \newpage
% \input{A-construct}

\end{document}